\documentclass[a4paper, 11pt]{article}

\usepackage{geometry}
\usepackage[utf8]{inputenc}
\usepackage[english]{babel}
\usepackage{lmodern}
\usepackage{microtype}

\usepackage{graphicx}
\usepackage{hyperref}
\usepackage{xcolor}
\usepackage{standalone}
\usepackage{tikz}
\usepackage{amsmath,amsfonts,amssymb}
\usepackage{amsthm}
\usepackage[capitalise]{cleveref}
\usepackage[shortlabels]{enumitem}
\usepackage[nocompress]{cite}

\usepackage{subcaption}
\usepackage{caption}

\usepackage{stmaryrd}

\usepackage[numbers]{natbib}

\newcommand{\N}{\ensuremath{\mathbb{N}}}
\newcommand{\Z}{\ensuremath{\mathbb{Z}}}

\newcommand{\R}{\ensuremath{\mathbb{R}}}

\newcommand{\Ann}{\ensuremath{\text{Ann}}}

\newcommand{\inner}[2]{{\langle #1, #2 \rangle}}
\newcommand{\verteq}{\rotatebox{90}{$\,=$}}
\newcommand{\Patt}[2]{\ensuremath{\mathcal{L}_{#2}(#1)}}

\renewcommand{\vec}[1]{\mathbf{#1}}

\newcommand{\SFT}[1]{\ensuremath{\mathcal{V}(#1)}}
\newcommand{\Valid}[1]{\ensuremath{\mathcal{V}(#1)}}
\newcommand{\Lang}[1]{\ensuremath{\mathcal{L}(#1)}}

\newcommand{\nelio}[4]{\ensuremath{\mathrm{Bin}}\left(
\raisebox{0.5mm}[1mm][0mm]{\scalebox{0.8}{\ensuremath{\begin{array}{cc}
#3 & #4\\
#1 & #2\\
\end{array}}
}}\hspace*{-1mm} \right)}

\newtheorem{theorem}{Theorem}
\newtheorem*{theorem*}{Theorem}
\newtheorem{lemma}[theorem]{Lemma}
\newtheorem*{lemma*}{Lemma}

\newtheorem{corollary}[theorem]{Corollary}
\crefname{corollary}{Corollary}{Corollaries}

\crefname{prop}{Property}{Properties}

\newtheorem{proposition}[theorem]{Proposition}
\crefname{proposition}{Proposition}{Propositions}

\theoremstyle{definition}

\theoremstyle{remark}

\crefname{remark}{Remark}{Remarks}

\crefname{ex}{Example}{Examples}

\newtheorem*{conjecture}{Conjecture}

\title{Decidability and Periodicity of Low Complexity Tilings}

\makeatletter
\renewcommand\@date{{%
  \vspace{-\baselineskip}%
  \vspace{-\baselineskip}%
  \large\centering
  \begin{tabular}{@{}c@{}}
    Jarkko Kari\textsuperscript{1} \\
    \normalsize jkari@utu.fi
  \end{tabular}%
  \quad and\quad
  \begin{tabular}{@{}c@{}}
    Etienne Moutot\textsuperscript{2}\\
    \normalsize etienne.moutot@math.cnrs.fr
  \end{tabular}

  \bigskip

  \normalsize

  \textsuperscript{1} University of Turku, Finland\par
  \textsuperscript{2} Aix-Marseille Université -- CNRS -- LIS, Marseille, France

  \bigskip
}}
\makeatother

\makeatletter
\def\keywords{\xdef\@thefnmark{}\@footnotetext}
\makeatother

\begin{document}

\maketitle

\begin{abstract}
    In this paper we study colorings (or tilings) of the two-dimensional grid $\Z^2$.
    A coloring is said to be valid with respect to a set $P$ of $n\times m$ rectangular patterns if all $n\times m$ sub-patterns of the coloring are in $P$.
    A coloring $c$ is said to be of low complexity with respect to a rectangle if there exist $m,n\in\N$ and a set $P$ of $n\times m$ rectangular patterns such that $c$ is valid with respect to $P$ and $|P|\leq nm$.
    Open since it was stated in 1997, Nivat's conjecture states that such a coloring is necessarily periodic.
    If Nivat's conjecture is true, all valid colorings with respect to $P$ such that $|P|\leq mn$ must be periodic.
    We prove that there exists at least one periodic coloring among the valid ones.
    We use this result to investigate the tiling problem, also known as the domino problem, which is well known to be undecidable in its full generality.
    However, we show that it is decidable in the low-complexity setting.
    Then, we use our result to show that Nivat's conjecture holds for uniformly recurrent configurations.
    These results also extend to other convex shapes in place of the rectangle.\\
    After that, we prove that the $nm$ bound is multiplicatively optimal for the decidability of the domino problem, as for all $\varepsilon>0$ it is undecidable to determine if there exists a valid coloring for a given $m,n\in \N$ and set of rectangular patterns $P$ of size $n\times m$ such that $|P|\leq (1+\varepsilon)nm$.
    We prove a slightly better bound in the case where $m=n$, as well as constructing aperiodic SFTs of pretty low complexity.\\
    This paper is an extended version of a paper published in STACS 2020 \cite{stacs}.
\end{abstract}

\section{Introduction}

The tiling problem, also known as the domino problem, asks whether the two-dimensional grid $\Z^2$
can be colored in a way that avoids a given finite collection of forbidden local patterns.
The problem is undecidable in its full generality. The undecidability relies on the fact
that there are \emph{aperiodic} systems of forbidden patterns that enforce any valid coloring to be non-periodic~\cite{berger}.

An example of such systems are Wang tiles: square tiles with colored edges that can be placed next to each other if their abutting edge are matching.
In other words the forbidden patterns are all pairs of tiles with non-matching edges.
A set of tiles is called aperiodic if all its valid tilings are non periodic.
In this context, the minimum size of the alphabet (or number of tiles) for a tileset to be aperiodic is know to be 11 \cite{11tiles}.
However, if instead of the number of tiles we are interested in the number of local patterns that can appear in the tilings, we do not know what is the minimal number (or function) that gives an aperiodic SFT.

In this paper we first consider
the low complexity setup where the number of allowed local patterns is small. More precisely, suppose we are given at most $nm$
legal rectangular patterns of size $n\times m$, and we want to know whether there exists a coloring of $\Z^2$
containing only legal $n\times m$ patterns. We prove that if such a coloring exists then also
a periodic coloring exists (Corollary~\ref{cor:corperiodic}). This further implies, using standard arguments, that in this setup
there is an algorithm to determine if the given patterns admit at least one coloring of the grid (Corollary~\ref{cor:cordecidable}).
The results also extend to other convex shapes in place of the rectangle (see Section~\ref{sec:conclusions}).

Then, we investigate what can happen if the complexity slightly increases.
In order to better understand the boundaries of the undecidability of the domino problem in terms of pattern complexity, we consider what we call the \emph{pretty low complexity} case, where we prove that the domino problem becomes undecidable again for several bounds on the size of the set of allowed patterns. This pretty low complexity setting was introduced in \cite{karigoles}.
We show that the previous $nm$ bound is multiplicatively optimal, that is that for all $\varepsilon>0$, it is undecidable to determine whether it is possible to color a bi-infinite grid only using patterns from a given set of at most $(1+\varepsilon)nm$ allowed patterns of size $n\times m$.
In the case where $m=n$, we prove a slightly better bound where $(1+\varepsilon)nm$ is replaced with $n^2+f(n)n$ with $f:\N\rightarrow\N$ any unbounded computable function (\cref{cor:goles_2}).
We also obtain a construction of pretty low aperiodic SFTs (\cref{cor:goles_3}).

We believe the low complexity setting has relevant applications.
There are numerous examples of processes in physics, chemistry and biology where
macroscopic patterns and regularities arise from simple microscopic interactions. Formation of
crystals and quasi-crystals is a good example where physical laws govern locally the attachments of
particles to each other. Predicting the structure of the crystal from its chemical composition is a
notoriously difficult problem (as already implied by the undecidability of the tiling problem) but if the
number of distinct local patterns of particle attachments is sufficiently low, our results indicate
that the situation may be easier to handle. For a good reference on quasicrystal and aperiodic order, see \cite{aperiodicorder}.

Our work is also motivated by \emph{Nivat's conjecture}~\cite{nivat},
an open problem concerning periodicity in low complexity colorings of the grid.
The conjecture claims the following: if a coloring of $\Z^2$ is such that, for some $n,m\in\N$,
the number of distinct $n \times  m$ patterns is at most $nm$, then the coloring is necessarily
periodic in some direction. If true, this conjecture directly implies a strong form of our periodicity
result: in the low complexity setting, not only a coloring exists that is periodic, but in fact
all valid colorings are periodic.
Our contribution to Nivat's conjecture is that we show that under the hypotheses of the conjecture,
the coloring must contain arbitrarily large periodic regions (Theorem~\ref{thm:periodic}).


\section{Preliminaries}

We denote $\llbracket n,m \rrbracket=\{n,n+1,\dots ,m\}$ for integers $n\leq m$, and for any positive integer $n$ we set
$\llbracket n\rrbracket=\llbracket 0,n-1 \rrbracket$. We index the columns and rows of the $n\times m$
rectangle $\llbracket n\rrbracket\times \llbracket m\rrbracket$ by $0,\dots ,n-1$ and $0,\dots ,m-1$, respectively.
The $n\times m$ rectangle at position $\vec{u}\in\Z^2$ of the two-dimensional grid is
$\vec{u}+\llbracket n\rrbracket\times \llbracket m\rrbracket\subseteq\Z^2$.

Let $A$ be a finite alphabet. A coloring $c\in A^{\Z^2}$ of the two-dimensional grid $\Z^2$ with elements of $A$ is called a (two-dimensional) \emph{configuration}. We use the notation $c_{\vec{n}}$ for the color $c(\vec{n})\in A$ of cell $\vec{n}\in\Z^2$.
For any $\vec{t}\in\Z^2$, the \emph{translation} $\tau^{\vec{t}}:A^{\Z^2}\longrightarrow A^{\Z^2}$ by $\vec{t}$ is defined by $\tau^{\vec{t}}(c)_{\vec{n}}=c_{\vec{n}-\vec{t}}$, for all $c\in A^{\Z^2}$ and all $\vec{n}\in\Z^2$. If $\tau^{\vec{t}}(c)=c$ for a
non-zero $\vec{t}\in\Z^2$, we say that $c$ is \emph{periodic} and that $\vec{t}$ is a \emph{vector of periodicity}. If there are
two linearly independent vectors of periodicity then $c$ is \emph{two-periodic}, and in this case
there are horizontal and vertical vectors of periodicity $(k,0)$ and $(0,k)$ for some $k> 0$, and consequently a vector of periodicity in every rational direction.

A \emph{finite pattern} is a coloring $p\in A^D$ of some
finite domain $D\subset \Z^d$. For a fixed $D$, we call such $p$ also a \emph{$D$-pattern}.
The set $[p]=\{c\in A^{\Z^2}\ |\ c|_{D}=p\}$ of configurations that contain pattern $p$ in domain $D$ is
the \emph{cylinder} determined by $p$. We say that a pattern $p$ \emph{appears} in configuration $c$, or that $c$ \emph{contains} pattern $p$,
 if some translate $\tau^{\vec{t}}(c)$ of $c$ is in $[p]$.
 For a fixed finite $D$, the set of $D$-patterns that appear in a configuration $c$ is denoted by $\Patt{c}{D}$, that is,
\[ \Patt{c}{D}=\{\tau^{\vec{t}}(c)|_{D}\ |\ \vec{t}\in\Z^2\ \}. \]
We denote by $\Lang{c}$ the set of all finite patterns that appear in $c$, i.e., the union of
$\Patt{c}{D}$ over all finite $D\subseteq\Z^2$.

We say that $c$ has \emph{low complexity\/} with respect to shape $D$ if $|\Patt{c}{D}|\leq |D|$,
and we call $c$ a \emph{low complexity configuration\/} if it has low complexity with respect to some finite $D$.
\begin{conjecture}[Maurice Nivat 1997~\cite{nivat}]
Let $c\in A^{\Z^2}$ be a two-dimensional configuration. If $c$ has low complexity with respect to some rectangle $D=\llbracket n\rrbracket\times \llbracket m\rrbracket$ then $c$ is periodic.
\end{conjecture}
\noindent
The analogous claim  in dimensions higher than two fails, as does an analogous claim in two dimensions
for many shapes other than rectangles~\cite{cassaigne}.

\subsection{Algebraic concepts}

Kari and Szabados introduced in~\cite{kariszabados} an algebraic approach to study low complexity configurations. The present paper heavily relies on this
technique. In this approach we replace the colors in $A$ by distinct integers, so that we assume $A\subseteq \Z$. We then express
a configuration $c\in A^{\Z^2}$ as a formal power series $c(x,y)$ over two variables $x$ and $y$ in which the coefficient of monomial $x^iy^j$ is $c_{i,j}$, for all $i,j\in\Z$. Note that the exponents of the variables range from $-\infty$ to $+\infty$. 
Note also that variables $x$ and $y$ in our power series and polynomials are treated only as ``position indicators'': in this work we never plug in any values to the variables.
In the following, polynomials may have negative powers of variables, that is, the polynomials we consider here and in the following are Laurent polynomials.
Let us denote by $\Z[x^{\pm 1}, y^{\pm 1}]$ and $\Z[[x^{\pm 1}, y^{\pm 1}]]$ the sets of such polynomials and power series, respectively.
We call a power series $c\in \Z[[x^{\pm 1}, y^{\pm 1}]]$ \emph{finitary} if its coefficients take only finitely many different values. Since we color the grid using finitely many colors, configurations are identified with finitary power series.

Multiplying a configuration $c\in \Z[[x^{\pm 1}, y^{\pm 1}]]$
by a monomial corresponds to translating it, and the periodicity of the configuration by vector $\vec{t}=(n,m)$ is then
equivalent to $(x^ny^m-1)c=0$, the zero power series. More generally, we say that a polynomial $f\in \Z[x^{\pm 1}, y^{\pm 1}]$ \emph{annihilates} power series $c$ if the formal product $fc$ is the zero power series. 

The set of polynomials that annihilates a power series is a Laurent polynomial ideal, and is denoted by
\[ \Ann(c) = \{ f\in \Z[x^{\pm 1}, y^{\pm 1}] ~|~ fc=0 \} .\]

It was observed in~\cite{kariszabados} that if a configuration
has low complexity with respect to some shape $D$  then it is annihilated by some non-zero polynomial $f\neq 0$.

\begin{lemma}[\cite{kariszabados}]
\label{th:low_complexity}
Let $c\in \Z[[x^{\pm 1}, y^{\pm 1}]]$ be a low complexity configuration.
Then $\Ann(c)$ contains a non-zero polynomial.
\end{lemma}

One of the main results of~\cite{kariszabados} states that if a
configuration $c$ is annihilated by a non-zero polynomial then it has annihilators of particularly nice form:
\begin{theorem}[\cite{kariszabados}]
  \label{th:decompo}
  Let $c\in\Z[[x^{\pm 1}, y^{\pm 1}]]$ be a configuration (a finitary power series) annihilated by some non-zero polynomial.
  Then there exist pairwise linearly independent $(i_1, j_1), \ldots, (i_m, j_m)\in\Z^2$ such that
  \[ (x^{i_1}y^{j_1} - 1) \cdots (x^{i_m}y^{j_m} - 1) \in \Ann(c) .\]
\end{theorem}
\noindent
Note that both Lemma~\ref{th:low_complexity} and Theorem~\ref{th:decompo} were proved in~\cite{kariszabados} for configurations $c\in A^{\Z^d}$
in arbitrary dimension $d$. In this work we only deal with two-dimensional configurations, so above we stated these results for $d=2$.

If $X\subseteq  A^{\Z^2}$ is a set of configurations, we denote by $\Ann(X)$ the set of Laurent polynomials that annihilate all elements of $X$.
We call $\Ann(X)$ the annihilator ideal of $X$.

\subsection{Dynamical systems concepts}
\label{sec:dynamical}

Cylinders $[p]$ are a base of a compact topology on $A^{\Z^2}$, namely the product of discrete topologies on $A$. See, for example,
the first few pages of~\cite{tullio}.
The topology is equivalently defined by a metric on $A^{\Z^2}$ where two configurations are close to each other if they agree
with each other on a large region around cell $\vec{0}$.

A subset $X$ of $A^{\Z^2}$ is a \emph{subshift} if it is closed in the topology
and closed under translations. Equivalently, every configuration $c$ that is not in $X$ contains a finite pattern $p$
that prevents it from being in $X$: no configuration that contains $p$ is in $X$.
We can then define subshifts using forbidden patterns as well:
for a set $F$ of finite patterns, define
\[ X_F=\{c\in A^{\Z^2}\ \mid  \Lang{c}\cap F=\emptyset\}, \]
the set of configurations that avoid all patterns in $F$. A set $X_F$ is a subshift, and every subshift is $X_F$ for some $F$.
If $X=X_F$ for some finite $F$ then $X$ is a \emph{subshift of finite type} (SFT).
For a subshift $X\subseteq A^{\Z^2}$ we denote by $\Patt{X}{D}=\cup_{c\in X} \Patt{c}{D}$ and $\Lang{X}=\cup_{c\in X} \Lang{c}$
the sets of $D$-patterns and all finite patterns that appear in
elements of $X$, respectively. Set $\Lang{X}$ is called the \emph{language} of the subshift.

Subshifts of finite type can be defined in terms of \emph{allowed patterns} as well. To do so we fix a finite domain $D\subseteq\Z^2$, and
take a set $P\subseteq A^D$ of allowed patterns with domain $D$. Forbidding all other $D$-patterns yields the SFT
\[
\SFT{P} = X_{A^D\setminus P}=\{c\in A^{\Z^2}\ \mid \Patt{c}{D}\subseteq P \},
\]
the set of configurations whose $D$-patterns are among $P$. We call elements of $\SFT{P}$ \emph{valid configurations} for $P$.

We call an SFT \emph{aperiodic} if it is non-empty but does not contain any periodic configurations.

The \emph{tiling problem} (aka the domino problem) is the decision problem that
asks whether a given SFT is empty, that is, whether there exists a configuration avoiding a given finite collection $P$ of forbidden finite patterns. Usually this question is asked in terms of so-called Wang tiles, but our formulation is equivalent.
The tiling problem is undecidable~\cite{berger}. 
An SFT is called \emph{aperiodic} if it is non-empty but does not contain any periodic configurations. It is significant that aperiodic SFTs exist~\cite{berger}, and in fact they must exist because of the undecidability of the tiling
problem~\cite{wang}. We recall the reason for this fact in the proof of Corollary~\ref{cor:cordecidable}.
It is also worth noting that
a two-dimensional SFT that contains a periodic configuration must also contain
a two-periodic configuration~\cite{robinson1971undecidability}.

Convergence of a sequence $c^{(1)}, c^{(2)}, \ldots$ of configurations to a configuration $c$
in our topology has the following simple meaning: For every cell $\vec{n}\in\Z^2$ we must have $c^{(i)}_{\vec{n}} = c_{\vec{n}}$ for all sufficiently large $i$. As usual, we denote then $c=\lim_{i\rightarrow\infty} c^{(i)}$. Note that if all $c^{(i)}$ are in a subshift $X$, so is the limit.
Compactness of space $A^{\Z^2}$ means that every sequence has a converging subsequence. In the proof of Theorem~\ref{thm:main}
in Section~\ref{sec:main} we frequently use this fact and extract converging subsequences from sequences of configurations.

The \emph{orbit} of configuration $c$ is the set ${\cal  O}(c) = \{\tau^{\vec{t}}(c)\ |\ \vec{t}\in\Z^2 \}$ that contains all translates of $c$.
The \emph{orbit closure} $\overline{{\cal  O}(c)}$ of $c$ is the topological closure of the orbit ${\cal  O}(c)$. It is a subshift, and in fact it is the intersection of all subshifts that contain $c$. The orbit closure $\overline{{\cal  O}(c)}$ can hence be called the subshift
generated by $c$.
In terms of finite patterns,
$c'\in \overline{{\cal  O}(c)}$ if and only if every finite pattern that appears in $c'$ appears also in $c$.
$\overline{{\cal  O}(c)}$ can be seen as the subshift containing all the translates of $c$ (its orbit) and all the limits of those translates.
Thus it can be different of ${\cal  O}(c)$: if $c$ is the configuration that with a black cell at the origin and white everywhere else, all the configurations of its orbit will contain a black cell, but at different positions; however its orbit closure contains the configuration with only white cells, as it is a limit of translations of $c$.

A configuration $c$ is called \emph{uniformly recurrent} if for every $c'\in \overline{{\cal  O}(c)}$ we have
$\overline{{\cal  O}(c')} = \overline{{\cal  O}(c)}$. This is equivalent to  $\overline{{\cal  O}(c)}$ being
a \emph{minimal subshift} in the sense that it has no proper non-empty subshifts inside it. A classical result by Birkhoff~\cite{birkhoff}
implies that every non-empty subshift contains a minimal subshift, so there is a uniformly recurrent configuration
in every non-empty subshift.

We use the notation $\inner{\vec{x}}{\vec{y}}$ for the inner product of vectors $\vec{x},\vec{y}\in \Z^2$.
For a nonzero vector $\vec{u}\in\Z^2\setminus\{\vec{0}\}$  we denote\[
H_{\vec{u}} = \{\vec{x}\in\Z^2\ |\ \inner{\vec{x}}{\vec{u}} < 0 \}
\]
for the discrete \emph{half plane} in direction $\vec{u}$. See Figure~\ref{fig:halfplanes}(a) for an illustration.
A subshift $X$ is \emph{deterministic} in direction $\vec{u}$ if for all $c,c'\in X$
\[
c|_{H_{\vec{u}}}=c'|_{H_{\vec{u}}} \Longrightarrow c=c',
\]
that is, if the contents of a configuration in the half plane $H_{\vec{u}}$ uniquely determines the contents in the rest of the cells.
Note that it is enough to verify that the value $c_{\vec{0}}$ on the boundary of the half plane is uniquely determined.
Indeed, if $c|_{H_{\vec{u}}}$ uniquely determines the line at its boundary, it is also true for all the translations of $c$, so the next line is also uniquely determined. By repeating this process the whole configuration is determined by $c|_{H_{\vec{u}}}$.
Moreover, by compactness, determinism in direction $\vec{u}$ implies that
there is a finite number $k$ such that already the contents  of a configuration in the discrete box
\[
B_{\vec{u}}^{k} = \{\vec{x}\in\Z^2\ |\ -k <  \inner{\vec{x}}{\vec{u}} < 0 \mbox{ and } -k <  \inner{\vec{x}}{\vec{u}^\bot} <  k \}
\]
are enough to uniquely determine the contents in cell ${\vec{0}}$,
where we denote by $\vec{u}^{\bot}$ a vector that is orthogonal to $\vec{u}$ and
has the same length as $\vec{u}$, e.g., $(n,m)^{\bot}=(m,-n)$. See Figure~\ref{fig:halfplanes}(b) for an illustration.

\begin{figure}[ht]%
  \centering
  \subcaptionbox{The discrete half plane $H_{\vec{u}}$ \label{fig:halfplanes:a}}{%
    \includegraphics[width=5cm]{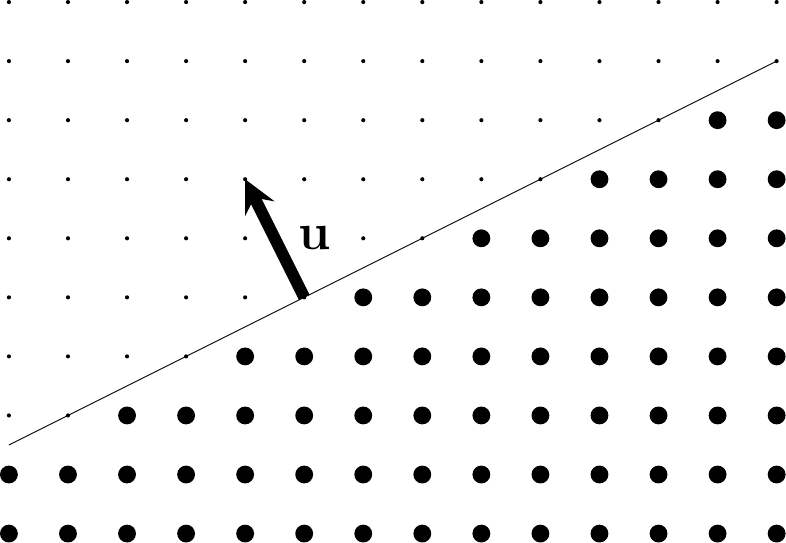}
  }%
  \qquad
  \qquad
  \subcaptionbox{The discrete box $B_{\vec{u}}^k$ with $k=10$. \label{fig:halfplanes:b}}{%
    \includegraphics[width=5cm]{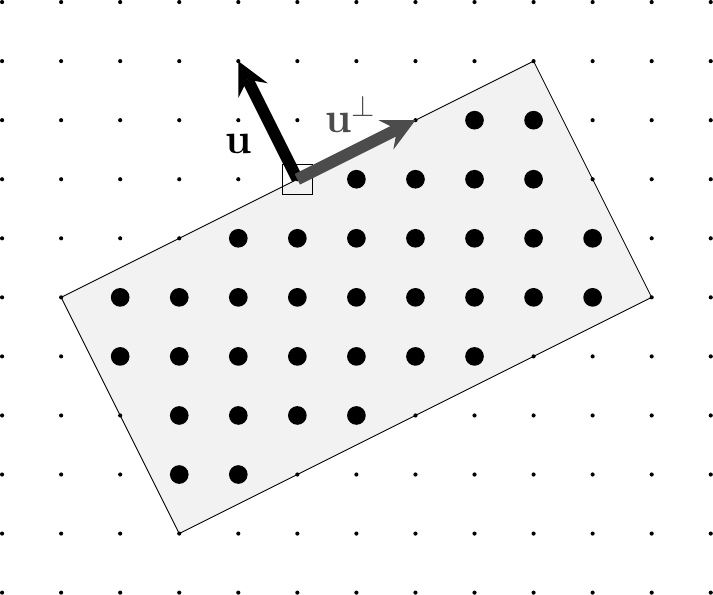}
  }%
  \caption{Discrete regions determined by vector $\vec{u}=(-1,2)$.}
  \label{fig:halfplanes}
\end{figure}

If $X$ is deterministic in directions $\vec{u}$ and  $-\vec{u}$ we say that $\vec{u}$ is a direction of \emph{two-sided} determinism.
If  $X$ is deterministic in direction $\vec{u}$ but not in direction $-\vec{u}$ we say that $\vec{u}$ is a direction of \emph{one-sided} determinism. Directions of two-sided determinism correspond to directions of expansivity in the symbolic dynamics literature. If $X$
is not deterministic in direction $\vec{u}$ we call $\vec{u}$ a \emph{direction of non-determinism}.
Finally, note that the concept of determinism in direction $\vec{u}$ only depends on the orientation of vector $\vec{u}$
and not on its magnitude.

\subsection{Wang tiles}

Two-dimensional SFTs are commonly studied in terms of \emph{Wang tiles}, and the first aperiodic SFTs were constructed and the undecidability of the domino problem
was originality proved in the Wang tile formalism. A Wang tile is a unit square tile with colored edges, represented as a
4-tuple $$a=(a_\uparrow, a_\rightarrow, a_\downarrow, a_\leftarrow)\in C^4$$ of colors of the north, the east, the south and the west edges of the tile, respectively, where $C$ is a
set of colors. (See Figure~\ref{fig:wang}.)
\begin{figure}[htb]
\begin{center}
\includegraphics[scale=0.38]{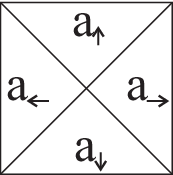}
\end{center}
\caption{A Wang tile $a=(a_\uparrow, a_\rightarrow, a_\downarrow, a_\leftarrow)$.}
\label{fig:wang}
\end{figure}
A Wang tile set $T$ is a finite set of Wang tiles. A Wang tile set $T$ defines a subshift of $T^{\Z^2}$, where
forbidden patterns are all the dominoes of two tiles that do not have the same color on their
abutting edges. We say that a configuration $c\in T^{\Z^2}$ is correctly tiled at position $(i,j)\in \Z^2$
if $c(i,j)$ matches with its four neighbors on the abutting edges so that
$$
\begin{array}{lcll}
c(i,j)_\uparrow &=& c(i,j+1)_\downarrow, &  \\
c(i,j)_\downarrow &=& c(i,j-1)_\uparrow, & \\
c(i,j)_\rightarrow &=& c(i+1,j)_\leftarrow  & \mbox{ and } \\
c(i,j)_\leftarrow &=& c(i-1,j)_\rightarrow.
\end{array}
$$
Otherwise there is a tiling error at position $(i,j)$.
We let
$$\Valid{T} = \{c\in T^{\Z^2}\ |\ \mbox{ $c$ is correctly tiled at every position $\vec{u}\in\Z^2$ } \}$$
be the set of valid tilings by tile set  $T$.
Clearly $\Valid{T}$ is an SFT, and in fact any given set $P\subseteq A^D$ of allowed patterns  can be effectively converted into an equivalent
Wang tile set $T$ so that $\Valid{T}$ and $\SFT{P}$ are conjugate, i.e., homeomorphic under a translation invariant homeomorphism.
In this sense Wang tiles capture the entire complexity of two-dimensional subshifts of finite type.
Note that we use the same notation $\Valid{T}$ and $\SFT{P}$ for the sets of valid tilings by a Wang tile set $T$ and of valid configurations under allowed patterns $P$,
respectively. This should not cause any confusion since it is always clear from the context whether we are talking about Wang tiles or allowed patterns.

The \emph{cartesian product} $T_1\times T_2\subseteq (C_1\times C_2)^4$ of Wang tile sets $T_1\subseteq C_1^4$ and $T_2\subseteq C_2^4$ is the Wang tile set
that contains for all $(a_\uparrow, a_\rightarrow, a_\downarrow, a_\leftarrow)\in T_1$ and $(b_\uparrow, b_\rightarrow, b_\downarrow, b_\leftarrow)\in T_2$
the  tile $( (a_\uparrow,b_\uparrow), (a_\rightarrow,b_\rightarrow), (a_\downarrow,b_\downarrow), (a_\leftarrow,b_\leftarrow))$. The ``sandwich'' tiles in $T_1\times T_2$
have hence two layers that tile the plane independently according to $T_1$ and $T_2$, respectively.

The results reported below are based on Berger's theorem, stating in the Wang tile formalism
the existence of aperiodic SFTs and the undecidability of the domino problem.

\begin{theorem}[R. Berger~\cite{berger}]
\begin{enumerate}
\item[(a)]
There exists a Wang tile set $T$ that is aperiodic, that is, such that $\Valid{T}$ is non-empty but does not contain any periodic configurations.
\item[(b)]
It is undecidable to determine for a given Wang tile set $T$ whether $\Valid{T}$ is empty or not.
\end{enumerate}
\end{theorem}


\section{Our results}
In this section, we sum up our main results, and the proofs will be given in later sections

\begin{theorem}
\label{thm:main}
Let $c$ be a two-dimensional configuration that has a non-trivial annihilator. Then $\overline{{\cal  O}(c)}$ contains a configuration
$c'$ such that $\overline{{\cal  O}(c')}$ has no direction of one-sided determinism.
\end{theorem}

\noindent
From this result, using a technique by Cyr and Kra~\cite{cyrkra}, we then obtain the second main result, stating that under the hypotheses of Nivat's conjecture, a configuration contains arbitrarily large periodic regions.

\begin{theorem}
\label{thm:periodic}
Let $c$ be a two-dimensional configuration that has low complexity with respect to a rectangle. Then $\overline{{\cal  O}(c)}$ contains a periodic configuration.
\end{theorem}

\noindent
These two theorems are proved in Sections~\ref{sec:main} and \ref{sec:periodic}, respectively. But let us first demonstrate
how these results imply relevant corollaries. First we consider SFTs defined in terms of allowed rectangular patterns.
Let $D=\llbracket n\rrbracket\times \llbracket m\rrbracket$ for some $m,n\in\N$.

\begin{corollary}
\label{cor:corperiodic}
Let $P\subseteq A^D$ be a set of $D$-patterns over alphabet $A$.
If $|P|\leq nm$ and $\SFT{P}\neq\emptyset$ then $\SFT{P}$ contains a periodic configuration.
\end{corollary}
\begin{proof}
Let $c\in \SFT{P}$ be arbitrary. By Theorem~\ref{thm:periodic} then, $\overline{{\cal  O}(c)}\subseteq \SFT{P}$ contains a periodic configuration.
\end{proof}

\begin{corollary}
\label{cor:cordecidable}
There is an algorithm that, given as input a set of $D$-patterns over a finite alphabet, with $|P|\leq nm$, determines whether $\SFT{P}\neq \emptyset$.
\end{corollary}
\begin{proof}
This is a classical argumentation by H.~Wang~\cite{wang}: there is a semi-algorithm to test if a given SFT is empty, and there is
a semi-algorithm to test if a given SFT contains a periodic configuration. 
Let us denote $P\subseteq A^D$ the set of $D$-patterns given as input.
Since $\SFT{P}$ is an SFT, we can execute both of these semi-algorithms on $\SFT{P}$. By Corollary~\ref{cor:corperiodic}, if $\SFT{P}\neq\emptyset$ then $\SFT{P}$ contains a periodic configuration. Hence, exactly one of these two semi-algorithms will return a positive answer.
\end{proof}

\noindent
The next corollary solves Nivat's conjecture for uniformly recurrent configurations.
\begin{corollary}
\label{cor:corminimal}
A uniformly recurrent configuration $c$ that has low complexity with respect to a rectangle is periodic.
\end{corollary}
\begin{proof}
Because $c$ has low complexity with respect to a rectangle then by Theorem~\ref{thm:periodic} there is a periodic
configuration $c'\in\overline{{\cal  O}(c)}$.
Because $\overline{{\cal O}(c')}$ contains only translates and limits of translates of $c'$,
all configurations in $\overline{{\cal  O}(c')}$ are periodic.
Finally,  because $c$ is uniformly recurrent we have $\overline{{\cal  O}(c)}=\overline{{\cal  O}(c')}$, which implies that
all elements of $\overline{{\cal  O}(c)}$, including $c$ itself, are periodic.
\end{proof}

\noindent
In Section~\ref{sec:conclusions} we briefly argue that all of these results remain true if the $n\times m$ rectangle is replaced by any convex discrete shape.

\bigskip

Our third main result shows that we are able to encode any set of Wang tiles into a pretty low complexity SFT.

\begin{theorem}
\label{thm:main_recoding}
Let $T$ be a given Wang tile set. One can effectively find positive integers $N$ and $k$ such that for the given $T$ and for any given $n\geq N$ and $m\geq 2$ one can effectively construct a set $P$ of binary rectangular patterns of size $n\times m$ such that
the cardinality of $P$ is at most $nm+k(n+m)$ and  $\SFT{P}$ contains a (periodic) tiling if and only if $\Valid{T}$ contains a
(periodic, resp.) configuration.
\end{theorem}

As a consequence, we are able to prove bounds on the complexity of SFTs for which the domino problem is undecidable.

\begin{corollary}
\label{cor:goles_1}
Let $f:\N\longrightarrow\N$ be a computable function, $f\not\in \mathcal{O}(1)$.
The following problem is undecidable for any fixed $m\geq 2$:
Given $n$ and a set $P$ of at most $nm+f(n)n$ binary rectangular patterns
of size $n\times m$, is $\SFT{P}$ empty~?
\end{corollary}

\begin{proof}
We many-one reduce the domino problem. Let $T$ be any given set of Wang tiles. Compute constants $N$ and $k$ of
Theorem~\ref{thm:main_recoding}. For $n=N, N+1, N+2,\dots$ compute $f(n)$ until number $n\geq N$ is found such that $f(n)\geq k+km/n$.
Because $f\not\in \mathcal{O}(1)$ such $n$ exists. Using Theorem~\ref{thm:main_recoding} construct a set $P$ of
at most $nm+k(n+m)\leq nm+f(n)n$ binary patterns of size $n\times m$. By Theorem~\ref{thm:main_recoding} tiles $T$ admit a valid tiling if
and only if $\SFT{P}$ is non-empty.
\end{proof}

Corollary~\ref{cor:goles_1} is stated for thin blocks of constant height $m$. It is also
worth to consider fat blocks, e.g., of square shape. By the analogous proof, using $m=n$ instead of constant $m$ we obtain the
following result where the additive term is almost linear in $n$.

\begin{corollary}
\label{cor:goles_2}
Let $f:\N\longrightarrow\N$ be a computable function, $f\not\in \mathcal{O}(1)$.
The following problem is undecidable: Given $n$ and a set $P$ of at most $n^2+f(n)n$ binary square patterns
of size $n\times n$, is $\SFT{P}$ empty~?
\end{corollary}
\begin{proof}
We proceed as in the proof of Corollary~\ref{cor:goles_1}, except that we choose $n$ such that $f(n)\geq 2k$. By Theorem~\ref{thm:main_recoding} we can effectively construct a set $P$ of at most $n^2+k(n+n)\leq n^2+f(n)n$ binary patterns of size $n\times n$ such that
$\SFT{P}$ is non-empty if and only if $T$ admits a valid tiling.
\end{proof}

In particular, for any real number $\varepsilon>0$
it is undecidable if a given set $P$ of at most $(1+\varepsilon)n^2$ square patterns of size $n\times n$ admit a valid configuration.

As usual, undecidability comes together with aperiodicity. We obtain pretty low complexity aperiodic SFTs.
\begin{corollary}
\label{cor:goles_3}
Let $f:\N\longrightarrow\N$ be a function, $f\not\in \mathcal{O}(1)$.
There exists $n$ and an aperiodic SFT $\SFT{P}$
where $P$ consists of at most $n^2+f(n)n$ binary square patterns
of size $n\times n$. Also, for every fixed height $m\geq 2$,
there exists a width $n$ and an aperiodic SFT $\SFT{P'}$
where $P'$ consists of at most $nm+f(n)n$ binary rectangular patterns
of size $n\times m$.
\end{corollary}
\begin{proof}
Let $T$ be an aperiodic Wang tile set. Let $N$ and $k$ be as in
Theorem~\ref{thm:main_recoding}, and let $n\in\N$ be such that $f(n)\geq 2k$.
By Theorem~\ref{thm:main_recoding} there is a collection $P$ of at most $n^2+k(n+n)\leq n^2+f(n)n$ binary $n\times n$
patterns such that $\SFT{P}$ is aperiodic. For fixed $m$, choosing $n$ such that $f(n)\geq k+km/n$
gives $P'$ in the second claim.
\end{proof}

\section{Removing one-sided determinism}
\label{sec:main}

In this section we prove Theorem~\ref{thm:main} by showing how we can ``remove'' one-sided directions of determinism
from subshifts with annihilators.

Let $c$ be a configuration over alphabet $A\subseteq\Z$ that has a non-trivial annihilator.
By Theorem~\ref{th:decompo} it has then an annihilator $\phi_1\cdots\phi_m$ where each $\phi_i$ is of the form
\begin{equation}
\label{eq:phi}
\phi_i=x^{n_i}y^{m_i}-1 \mbox{ for some } \vec{v}_i=(n_i,m_i)\in\Z^2.
\end{equation}
Moreover, vectors $\vec{v}_i$ can be chosen pairwise linearly independent, that is, in different directions. If $m=0$ it means that $c=0$, therefore we may assume $m\geq 1$.

Denote $X=\overline{{\cal  O}(c)}$, the subshift generated by $c$. A polynomial that annihilates $c$ annihilates all
elements of $X$, because they only have local patterns that already appear in $c$. It is easy to see that $X$ can only be non-deterministic in
a direction that is perpendicular to one of the  directions $\vec{v}_i$ of the polynomials $\phi_i$:

\begin{proposition}
\label{prop:detann}
Let $c$ be a configuration annihilated by $\phi_1\cdots\phi_m$ where each $\phi_i$ is of the form \emph{(\ref{eq:phi})}.
Let $\vec{u}\in\Z^2$ be a direction that is not perpendicular to $\vec{v}_i$ for any $i\in\{1,\ldots, m\}$. Then $X=\overline{{\cal  O}(c)}$
is deterministic in direction $\vec{u}$.
\end{proposition}

\begin{proof}
Suppose $X$ is not deterministic in direction $\vec{u}$. By definition, there exist $d,e\in X$ such that $d\neq e$ but
$d|_{H_{\vec{u}}}=e|_{H_{\vec{u}}}$. Denote $\Delta=d-e$. Because $\Delta\neq 0$  but  $\phi_1\cdots\phi_m\cdot \Delta =0$, for some
$i$ we have $\phi_1\cdots\phi_{i-1}\cdot \Delta \neq 0$ and $\phi_1\cdots\phi_{i}\cdot \Delta =0$.
Denote $\Delta'=\phi_1\cdots\phi_{i-1}\cdot \Delta$. Because $\phi_i\cdot\Delta'=0$, configuration $\Delta'$ is periodic in direction $\vec{v}_i$.
But because $\Delta$ is zero in the half plane ${H_{\vec{u}}}$, also $\Delta'$
is zero in some translate $H'={H_{\vec{u}}}-\vec{t}$ of the half plane.
Since the periodicity vector $\vec{v}_i$ of $\Delta'$ is not perpendicular to $\vec{u}$, the periodicity transmits the values 0 from
the region $H'$ to the entire $\Z^2$. Hence $\Delta'=0$, a contradiction.
\end{proof}

Let $\vec{u}\in\Z^2$ be a one-sided direction of determinism of  $X$.
In other words, $\vec{u}$ is a direction of determinism but $-\vec{u}$ is not. By the proposition above, $\vec{u}$ is perpendicular
to some $\vec{v}_i$. Without loss of generality, we may assume $i=1$.
We denote $\phi=\phi_1$ and $\vec{v}=\vec{v}_1$.

Let $k$ be such that the contents of the discrete box $B=B_{\vec{u}}^{k}$ determine the content of cell $\vec{0}$, that is, for $d,e\in X$
\begin{equation}
\label{eq:detshape}
d|_B=e|_B \Longrightarrow d_{\vec{0}}=e_{\vec{0}}.
\end{equation}
As pointed out in Section~\ref{sec:dynamical}, any sufficiently large $k$ can be used.
We can choose $k$ so that $k>|\inner{\vec{u}^\bot}{\vec{v}}|$.
To shorten notations, let us also denote $H=H_{-\vec{u}}$.

\begin{lemma}
\label{lem:apulemma}
For any $d,e\in X$ such that $\phi d=\phi e$ holds:
\[
d|_B=e|_B \Longrightarrow d|_H=e|_H.
\]
\end{lemma}

\begin{proof}
Let $d,e\in X$ be such that $\phi d=\phi e$ and $d|_B=e|_B$.
Denote $\Delta=d-e$. Then $\phi\Delta=0$ and $\Delta|_B=0$.
Property $\phi\Delta=0$ means that $\Delta$ has periodicity vector $\vec{v}$, so this
periodicity transmits values 0 from the region $B$ to the stripe
\[
S=\bigcup_{i\in\Z}  (B+i\vec{v}) = \{\vec{x}\in\Z^2\ |\ -k< \inner{\vec{x}}{\vec{u}} < 0\},
\]
See Figure~\ref{fig:regions} for an illustration of the regions $H$, $B$ and $S$.
As $\Delta|_S=0$, we have that $d|_S=e|_S$. Applying (\ref{eq:detshape}) on suitable translates of $d$ and $e$ allows us to conclude that
$d|_H=e|_H$.
\end{proof}

\begin{figure}[ht]
  \centering
  \includegraphics[width=5.5cm]{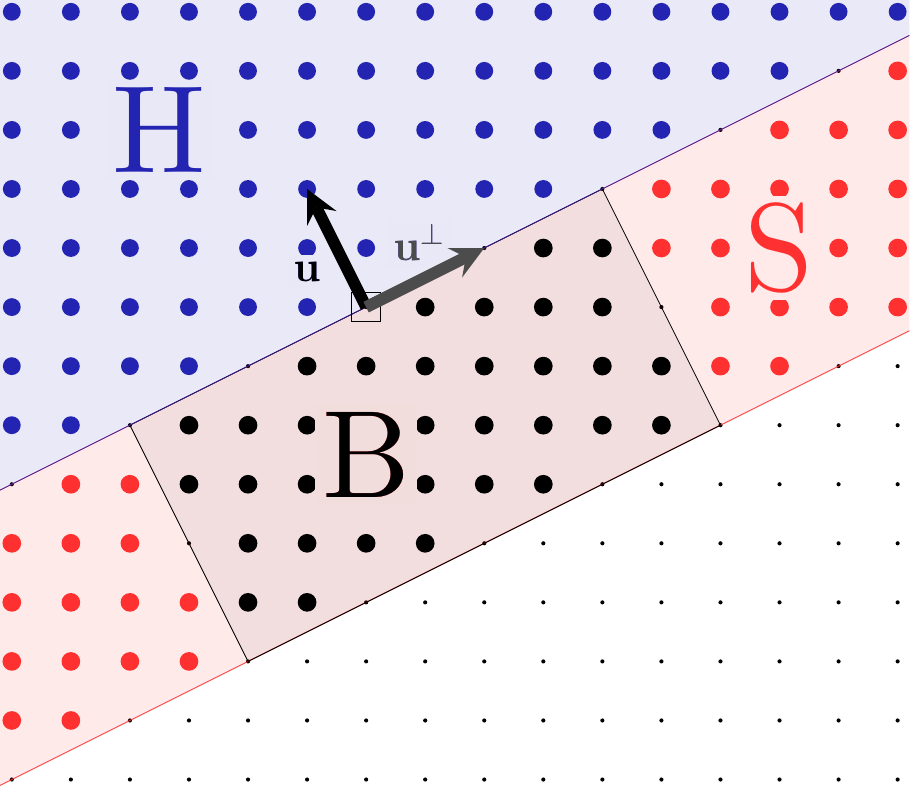}
  \caption{Discrete regions $H=H_{-\vec{u}}$, $B=B_{\vec{u}}^{k}$ and $S$ in the proof of Lemma~\ref{lem:apulemma}. In the illustration $\vec{u}=(-1,2)$ and $k=10$.}
  \label{fig:regions}
\end{figure}

A reason to prove the lemma above is the following corollary, stating that $X$ can only contain a bounded number of configurations that have
the same product with $\phi$:

\begin{corollary}
\label{cor:bounded}
Let $c_1,\ldots ,c_n\in X$ be pairwise distinct. If $\phi c_1=\cdots =\phi c_n$ then $n\leq |A|^{|B|}$.
\end{corollary}

\begin{proof}
Let $H'=H-\vec{t}$, for $\vec{t}\in\Z^2$,  be a translate of the half plane $H=H_{-\vec{u}}$ such that $c_1,\ldots ,c_n$ are
pairwise different on $H'$. Consider the translated configurations $d_i=\tau^{\vec{t}}(c_i)$.
We have that $d_i\in X$ are pairwise different on $H$ and  $\phi d_1=\cdots =\phi d_n$.
By Lemma~\ref{lem:apulemma}, configurations $d_i$ must be pairwise different on domain $B$. There are only $|A|^{|B|}$ different patterns
in domain $B$.
\end{proof}

Let $c_1,\ldots ,c_n\in X$ be pairwise distinct such that $\phi c_1=\cdots =\phi c_n$, with $n$
as large as possible. By Corollary~\ref{cor:bounded} such a maximal $n$ exists.
Let us repeatedly translate the configurations $c_i$ by $\tau^{\vec{u}}$ and take a limit: by compactness there exists
$n_1<n_2<n_3\ldots$ such that
\[
d_i=\lim_{j\rightarrow \infty} \tau^{n_j\vec{u}}(c_i)
\]
exists for all $i\in\{1,\ldots, n\}$.
Configurations $d_i\in X$ inherit the following properties from $c_i$:

\begin{lemma}
\label{lem:dlemma}
Let $d_1,\ldots ,d_n$ be defined as above. Then
\begin{enumerate}
\item[(a)]  $\phi d_1=\cdots =\phi d_n$, and
\item[(b)] Configurations $d_i$ are pairwise different on translated discrete boxes $B'=B-\vec{t}$ for all $\vec{t}\in\Z^2$.
\end{enumerate}
\end{lemma}

\begin{proof} Let $i_1,i_2\in\{1,\ldots ,n\}$ be arbitrary, $i_1\neq i_2$.
\medskip

\noindent
(a) Because $\phi c_{i_1} = \phi c_{i_2}$ we have, for any $n\in\N$,
\[\phi \tau^{n\vec{u}}(c_{i_1}) = \tau^{n\vec{u}}(\phi c_{i_1}) = \tau^{n\vec{u}}(\phi c_{i_2}) = \phi \tau^{n\vec{u}}(c_{i_2}).\]
Function $c\mapsto \phi c$
is continuous in the topology so
\[\phi d_{i_1} =\phi \lim_{j\rightarrow \infty} \tau^{n_j\vec{u}}(c_{i_1}) = \lim_{j\rightarrow \infty} \phi\tau^{n_j\vec{u}}(c_{i_1})
= \lim_{j\rightarrow \infty} \phi\tau^{n_j\vec{u}}(c_{i_2}) = \phi \lim_{j\rightarrow \infty} \tau^{n_j\vec{u}}(c_{i_2}) = \phi d_{i_2}.\]

\medskip

\noindent
(b) Let $B'=B-\vec{t}$ for some $\vec{t}\in\Z^2$.
Suppose $d_{i_1}|_{B'} = d_{i_2}|_{B'}$. By the definition of convergence, for all sufficiently large $j$ we have
$\tau^{n_j\vec{u}}(c_{i_1})|_{B'} = \tau^{n_j\vec{u}}(c_{i_2})|_{B'}$. This is equivalent to
$\tau^{n_j\vec{u}+\vec{t}}(c_{i_1})|_{B} = \tau^{n_j\vec{u}+\vec{t}}(c_{i_2})|_{B}$.
By Lemma~\ref{lem:apulemma} then also $\tau^{n_j\vec{u}+\vec{t}}(c_{i_1})|_{H} = \tau^{n_j\vec{u}+\vec{t}}(c_{i_2})|_{H}$
where $H=H_{-\vec{u}}$. This means that for all sufficiently large
$j$ the configurations $c_{i_1}$ and $c_{i_2}$ are identical on the domain $H-n_j\vec{u}-\vec{t}$. But these domains cover the whole $\Z^2$
as $j\longrightarrow\infty$ so that $c_{i_1}=c_{i_2}$, a contradiction.
\end{proof}

Now we pick one of the configurations $d_i$ and consider its orbit closure. Choose $d=d_1$ and
set $Y=\overline{{\cal  O}(d)}$. Then $Y\subseteq X$. Any direction of determinism in $X$
is also a direction of determinism in $Y$. Indeed, this is trivially true for any subset of $X$.
But, in addition, we have the following:

\begin{lemma}
\label{lem:onesidedremoved}
Subshift $Y$ is deterministic in direction $-\vec{u}$.
\end{lemma}

\begin{proof}
Suppose the contrary: there exist configurations $x,y\in Y$ such that $x\neq y$ but $x|_H=y|_H$ where, as usual,
$H=H_{-\vec{u}}$.
In the following we construct $n+1$ configurations in $X$ that have the same product with $\phi$, which
contradicts the choice of $n$ as the maximum number of such configurations.

By the definition of $Y$ all elements of $Y$ are limits of sequences of translates of $d=d_1$, that is, there are translations
$\tau_1,\tau_2,\ldots$ such that $x=\lim_{i\rightarrow\infty} \tau_i(d)$, and translations
$\sigma_1,\sigma_2,\ldots$ such that $y=\lim_{i\rightarrow\infty} \sigma_i(d)$.
Apply the translations $\tau_1,\tau_2,\ldots$ on configurations $d_1,\ldots ,d_n$, and take jointly converging subsequences: by compactness there are $k_1<k_2<\ldots$ such that
\[
e_i=\lim_{j\rightarrow\infty} \tau_{k_j}(d_i)
\]
exists for all $i\in\{1,\ldots ,n\}$. Here, clearly, $e_1=x$.

\medskip

Let us prove that $e_1,\ldots ,e_n$ and $y$ are $n+1$ configurations that
(i) have the same product with $\phi$, and (ii) are pairwise distinct. This contradicts the choice of $n$ as the maximum number of
such configurations, and thus completes the proof.
\medskip

\noindent
\begin{enumerate}
\item First, $\phi x=\phi y$: Because $x|_H=y|_H$ we have $\phi x|_{H-\vec{t}} = \phi y|_{H-\vec{t}}$ for some $\vec{t}\in\Z^2$.
  Consider $c'=\tau^{\vec{t}}(\phi x-\phi y)$, so that $c'|_H=0$.
  As $\phi_2\cdots \phi_m$ annihilates $\phi x$ and $\phi y$, it also annihilates $c'$.
  An application of Proposition~\ref{prop:detann} on configuration $c'$ in place of $c$
   shows that $\overline{{\cal  O}(c')}$ is deterministic in direction $-\vec{u}$.
   (Note that $-\vec{u}$ is not perpendicular to $\vec{v}_j$ for any $j\neq 1$, because
  $\vec{v}_1$ and $\vec{v}_j$ are not parallel and $-\vec{u}$ is perpendicular to $\vec{v}_1$.)
  Due to the determinism,  $c'|_H=0$ implies that $c'=0$, that is, $\phi x=\phi y$.

  Second, $\phi e_{i_1}=\phi e_{i_2}$ for all $i_1,i_2\in\{1,\ldots ,n\}$: By Lemma~\ref{lem:dlemma} we know that $\phi d_{i_1}=\phi d_{i_2}$. By continuity of the function $c\mapsto \phi c$  we then have
  \[
  \begin{array}{rc}
  \phi e_{i_1} = \phi \lim_{j\rightarrow\infty} \tau_{k_j}(d_{i_1}) = \lim_{j\rightarrow\infty} \phi \tau_{k_j}(d_{i_1}) =&
  \lim_{j\rightarrow\infty} \tau_{k_j}(\phi d_{i_1}) \\
  &\mbox{\verteq}\vspace*{-2mm}\\
  \phi e_{i_2} = \phi \lim_{j\rightarrow\infty} \tau_{k_j}(d_{i_2}) = \lim_{j\rightarrow\infty} \phi \tau_{k_j}(d_{i_2}) =&
  \lim_{j\rightarrow\infty} \tau_{k_j}(\phi d_{i_2})
  \end{array}
  \]
  Because $e_1=x$, we have shown that $e_1,\ldots ,e_n$ and $y$ all have the same product with $\phi$.


\item Pairwise distinctness: First, $y$ and $e_1=x$ are distinct by the initial choice of $x$ and $y$.
  Next, let $i_1,i_2\in\{1,\ldots ,n\}$ be such that $i_1\neq i_2$. Let $\vec{t}\in \Z^2$ be arbitrary and consider the translated
  discrete box $B'=B-\vec{t}$. By
  Lemma~\ref{lem:dlemma}(b) we have $\tau_{k_j}(d_{i_1})|_{B'}\neq \tau_{k_j}(d_{i_2})|_{B'}$ for all $j\in \N$, so taking the
  limit as $j\longrightarrow\infty$ gives $e_{i_1}|_{B'}\neq e_{i_2}|_{B'}$. This proves that $e_{i_1}\neq e_{i_2}$. Moreover, by
  taking $\vec{t}$ such that $B'\subseteq H$ we see that $y|_{B'}=x|_{B'}=e_1|_{B'}\neq e_i|_{B'}$ for $i\geq 2$, so that
  $y$ is also distinct from all $e_i$ with $i\geq 2$.
\end{enumerate}
\vspace{-1em}
\end{proof}

\noindent
The following proposition captures the result established above.
\begin{proposition}
\label{prop:main}
Let $c$ be a configuration with a non-trivial annihilator. If $\vec{u}$
is a  one-sided  direction of determinism in $\overline{{\cal  O}(c)}$ then
there is a configuration $d\in \overline{{\cal  O}(c)}$ such that  $\vec{u}$
is a  two-sided  direction of determinism in  $\overline{{\cal  O}(d)}$.
\end{proposition}

\begin{proof}
  Let $c$ be a configuration with a non-trivial annihilator and $\vec u$ a one-sided  direction of determinism in $\overline{{\cal  O}(c)}$.
  Then, consider $c_1, \hdots, c_n$ as in \cref{cor:bounded} and $n$ as large as possible.
  Then, by taking \[ d = \lim_{j\rightarrow \infty} \tau^{n_j\vec{u}}(c_1) ,\]
  \cref{lem:onesidedremoved} ensures that $Y=\overline{{\cal  O}(d)}$ is deterministic both in directions $\vec u$ and $- \vec u$, as none of the limits used change the determinism along $\vec u$.
\end{proof}

\noindent
Now we are ready to prove Theorem~\ref{thm:main}.

\begin{proof}[Proof of Theorem~\ref{thm:main}]
Let $c$ be a two-dimensional configuration that has a non-trivial annihilator. Every non-empty subshift contains a minimal subshift~\cite{birkhoff}, and hence there
is a uniformly recurrent configuration $c'\in\overline{{\cal  O}(c)}$. If $\overline{{\cal  O}(c')}$ has a one-sided direction of determinism $\vec{u}$, we can apply
Proposition~\ref{prop:main} on $c'$ and find $d\in\overline{{\cal  O}(c')}$ such that $\vec{u}$ is a two-sided direction of determinism in  $\overline{{\cal  O}(d)}$.
But because $c'$ is uniformly recurrent, $\overline{{\cal  O}(d)}=\overline{{\cal  O}(c')}$, a contradiction.
%
%
%
\end{proof}


\section{Periodicity in low complexity subshifts}
\label{sec:periodic}

In this section we prove Theorem~\ref{thm:periodic}. Every non-empty subshift contains  a uniformly recurrent configuration, so we can safely assume that $c$ is uniformly recurrent.

Our proof of Theorem~\ref{thm:periodic} splits in two cases based on Theorem~\ref{thm:main}: either $\overline{{\cal  O}(c)}$ is deterministic in all directions or for some $\vec{u}$ it is non-deterministic in both directions $\vec{u}$ and $-\vec{u}$. The first case is handled by the following well-known corollary from a theorem of Boyle and Lind~\cite{Boyle_Lind}:

\begin{proposition}
\label{prop:case1}
  A configuration $c$ is two-periodic if and only if $\overline{{\cal  O}(c)}$ is deterministic in all directions.
\end{proposition}

For the second case we apply the technique by
Cyr and Kra~\cite{cyrkra}. This technique was also used in~\cite{szabados} to address Nivat's conjecture.
It is possible to use a direct combination of lemmas from~\cite{cyrkra} or~\cite{szabados} to prove the following:

\begin{proposition}
\label{prop:case2}
Let $c$ be a two-dimensional uniformly recurrent configuration that has low complexity with respect to a rectangle.
If for some $\vec{u}$ both $\vec{u}$ and $-\vec{u}$ are directions of non-determinism in $\overline{{\cal  O}(c)}$ then $c$ is periodic in a direction perpendicular to $\vec{u}$.
\end{proposition}

We will prove this proposition below using lemmas from~\cite{szabados}.
We first recall some definitions, adjusted to our terminology. Let $D\subseteq \Z^2$ be non-empty and let
$\vec{u}\in\Z^2\setminus\{\vec{0}\}$. The \emph{edge} $E_{\vec{u}}(D)$ of $D$ in direction $\vec{u}$ consists of the cells in $D$
that are extremal in the direction $\vec u$:
\[
E_{\vec{u}}(D) = \{ \vec{v}\in D\ |\ \forall \vec{x}\in D\ \inner{\vec{x}}{\vec{u}} \leq \inner{\vec{v}}{\vec{u}} \}.
\]
We call $D$ \emph{convex} if $D=C\cap\Z^2$ for a convex subset $C\subseteq\R^2$ of the real plane.
For $D,E\subseteq \Z^2$ we say that $D$ \emph{fits} in $E$
if $D+\vec{t}\subseteq E$ for some $\vec{t}\in\Z^2$.

The (closed) \emph{stripe} of width $k$ perpendicular to $\vec{u}$ is the set
\[
S_{\vec{u}}^k = \{\vec{x}\in\Z^2\ |\ -k< \inner{\vec{x}}{\vec{u}} \leq 0\}.
\]
Consider the stripe $S=S_{\vec{u}}^k$.
The reader can refer to \cref{fig:regions} for an illustration of a closed stripe, the only difference being the inclusion of the upper boundary of $S$.
Clearly its edge $E_{\vec{u}}(S)$ in direction $\vec{u}$ is the discrete line $\Z^2\cap L$ where $L\subseteq\R^2$
is the real line through $\vec{0}$ that is perpendicular to $\vec{u}$. The \emph{interior} $S^\circ$ of $S$ is $S\setminus E_{\vec{u}}(S)$, that is, $S^\circ = \{\vec{x}\in\Z^2\ |\ -k< \inner{\vec{x}}{\vec{u}} < 0\}$.

A central concept from~\cite{cyrkra,szabados} is the following.  Let $c$ be a configuration and let $\vec{u}\in\Z^2\setminus\{\vec{0}\}$ be a direction. Recall that $\Patt{c}{D}$
 denotes the set of $D$-patterns that appear in $c$. A finite discrete convex set $D\subseteq \Z^2$ is called \emph{$\vec{u}$-balanced in $c$} if the following three conditions are satisfied, where
we denote $E=E_{\vec{u}}(D)$ for the edge of $D$ in direction $\vec{u}$:
\begin{enumerate}
\item[(i)] $|\Patt{c}{D}|\leq |D|$,
 \item[(ii)] $|\Patt{c}{D}| < |\Patt{c}{D\setminus E}|+|E|$, and
\item[(iii)] $|D\cap L|\geq |E|-1$ for every line $L$ perpendicular to $\vec{u}$ such that $D\cap L\neq\emptyset$.
\end{enumerate}
The first condition states that $c$ has low complexity with respect to shape $D$. The second condition implies that
there are fewer than $|E|$ different $(D\setminus E)$-patterns in $c$ that can be extended in more than one way into a $D$-pattern of $c$.
The last condition states that the edge $E$ is nearly the shortest among the parallel cuts across $D$.

\begin{lemma}[Lemma 2 in~\cite{szabados}]
\label{lem:Michal2}
Let $c$ be a two-dimensional configuration that has low complexity with respect to a rectangle, and let $\vec{u}\in\Z^2\setminus\{\vec{0}\}$.
Then $c$ has a $\vec{u}$-balanced or a $(-\vec{u}$)-balanced set $D \subseteq \Z^2$.
\end{lemma}

A crucial observation in~\cite{cyrkra} connects balanced sets and non-determinism to periodicity. This leads to the following
statement.

\begin{lemma}[Lemma 4 in~\cite{szabados}]
\label{lem:Michal4}
Let $d$ be a two-dimensional configuration and let $\vec{u}\in\Z^2\setminus\{\vec{0}\}$ be such that
$d$ admits a $\vec{u}$-balanced set $D \subseteq \Z^2$.
Assume there is a configuration $e\in\overline{{\cal  O}(d)}$ and a stripe  $S=S_{\vec{u}}^k$ perpendicular to $\vec{u}$
such that $D$ fits in $S$ and $d|_{S^\circ}=e|_{S^\circ}$ but $d|_{S}\neq e|_{S}$. Then $d$ is periodic in direction perpendicular to $\vec{u}$.
\end{lemma}

With these we can prove Proposition~\ref{prop:case2}.

\begin{proof}[Proof of Proposition~\ref{prop:case2}]
Let $c$ be a two-dimensional uniformly recurrent configuration that has low complexity with respect to a rectangle.
Let $\vec{u}$ be such that both $\vec{u}$ and $-\vec{u}$ are directions of non-determinism in $\overline{{\cal  O}(c)}$.
By Lemma~\ref{lem:Michal2} configuration  $c$ admits a
 $\vec{u}$-balanced or a $(-\vec{u}$)-balanced set $D \subseteq \Z^2$. Without loss of generality, assume that $D$ is
 $\vec{u}$-balanced in $c$. As $\overline{{\cal  O}(c)}$ is non-deterministic in direction $\vec{u}$, there are
 configurations $d,e\in \overline{{\cal  O}(c)}$ such that $d|_{H_{\vec{u}}}=e|_{H_{\vec{u}}}$ but $d_{(0,0)}\neq e_{(0,0)}$.
 Because $c$ is uniformly recurrent, exactly the
 same finite patterns appear in $d$ as in $c$. This means that $D$ is
 $\vec{u}$-balanced also in $d$. From the uniform recurrence of $c$ we also get that $e\in \overline{{\cal  O}(d)}$.
 Pick any $k$ large enough so that $D$ fits in the stripe  $S=S_{\vec{u}}^k$. Because $\vec{0}\in S$ and $S^\circ\subseteq H_{\vec{u}}$,
 the conditions in Lemma~\ref{lem:Michal4} are met. By the lemma, configuration $d$ is $\vec{p}$-periodic for some $\vec{p}$ that is perpendicular to $\vec{u}$. Because $d$ has the same finite patterns as $c$, it follows that $c$ cannot contain a pattern that breaks period $\vec{p}$.
 So $c$ is also $\vec{p}$-periodic.
 \end{proof}

 Now Theorem~\ref{thm:periodic} follows from Propositions~\ref{prop:case1} and \ref{prop:case2}, using Theorem~\ref{thm:main} and the fact
 that every subshift contains a uniformly recurrent configuration.

 \begin{proof}[Proof of Theorem~\ref{thm:periodic}] Let $c$ be a two-dimensional configuration that has low complexity with respect to a rectangle.
 Replacing $c$ by a uniformly recurrent element of $\overline{{\cal  O}(c)}$, we may assume that $c$ is uniformly recurrent.
 Since $c$ is a low-complexity configuration, by Lemma~\ref{th:low_complexity} it has a non-trivial annihilator. By Theorem~\ref{thm:main}
 there exists $c'\in \overline{{\cal  O}(c)}$ such that $\overline{{\cal  O}(c')}$ has no direction of one-sided determinism.
 If all directions are deterministic in $\overline{{\cal  O}(c')}$, it follows from Proposition~\ref{prop:case1} that $c'$ is two-periodic.
 Otherwise there is a direction $\vec{u}$ such that both $\vec{u}$ and $-\vec{u}$ are directions of non-determinism in $\overline{{\cal  O}(c')}$. Now it follows from Proposition~\ref{prop:case2} that $c'$ is periodic.
 \end{proof}


\section{Recoding Wang tiles}
\label{sec:recoding}

In this section we prove \cref{thm:main_recoding}.
We convert an arbitrary Wang tile set $T$ into a pretty small set $P$ of binary rectangular allowed patterns that is equivalent to $T$
in the sense that $P$ admits a (periodic) configuration if and only if $T$ admits a (resp. periodic) configuration. Configurations valid for P have bits $1$
sparsely positioned so that each bit $1$ represents a single Wang tile of a valid tiling, and
the relative positions of bits $1$ uniquely identify the corresponding Wang tiles. Allowed patterns in $P$ are restricted so that only
matching Wang tiles are allowed next to each other.
We detail this construction in the next pages.

So let $T$ be a given finite set of Wang tiles. We first modify the set to make sure that no tile matches itself as its neighbor.
This is easy to enforce by making two copies of $T$ and forcing the copies be used alternatingly on even and odd cells. More precisely,
we replace $T$ by the cartesian product $T\times\{\mbox{{\sc even}, {\sc odd}}\}$ where {\sc even} has color $0$ on its north and
east sides and color $1$ on south and west, while in  {\sc odd} the colors are reversed. The {\sc even}/{\sc odd} -components of tiles form an infinite checkerboard tiling of the plane. The new tile set admits a (periodic) tiling if and only if $T$ admits a (periodic, resp.) tiling.

From now on we assume that no tile of $T$ matches in color with itself. Let $t=|T|$ be the number of tiles, and denote
$$S=\{2^j-1\ |\ j=0,1,\dots ,t-1\}$$ and $s=2^{t-1}$. The set $S\subseteq\llbracket s\rrbracket $ has the property that for $a,b\in S$, $a\neq b$, the difference $a-b$
uniquely identifies both $a$ and $b$. The proof of this fact is easy.
\begin{lemma}
\label{lem:unique}
For $a_1,a_2,b_1,b_2\in S$, if $a_1-b_1 = a_2-b_2 \neq 0$ then $a_1=a_2$ and $b_1=b_2$.
\end{lemma}
Fix a bijection $\alpha:T\longrightarrow S$. In our coding tile $t$
will be represented as a horizontal sequence of $s$ bits where bit number $\alpha(t)$ is set to $1$ and all other bits are $0$'s.

Choose $N=3s$ and
fix $m\geq 2$ and $n\geq N$, the dimensions of the rectangular patterns considered, and define $$D=\llbracket n\rrbracket \times \llbracket m\rrbracket.$$
Denote $n'=n-s$ and $m'=m-1$. In our coding of a Wang tiling we paste to position
$(i\cdot n', j\cdot m')$
the bit sequence representing the Wang tile in position $(i,j)$.
A configuration $c\in T^{\Z^2}$
is then represented as a binary configuration $\beta(c)\in \{0,1\}^{\Z^2}$ where for all $(i,j)\in\Z^2$, tile $c(i,j)$ contributes
symbol $1$ in position $(in'+\alpha(c(i,j)), jm')$. All positions without a contribution from any tile of $c$ have value $0$.

In $\beta(c)$ all symbols $1$ appear in the intersections of vertical strips
$$
V_i=(in'+\llbracket s\rrbracket) \times \Z
$$
and horizontal strips
$$
H_j=\Z\times \{jm'\},
$$
for $i,j\in\Z$. There is exactly one symbol $1$ in each intersection $I_{i,j}=V_i\cap H_j$, representing the Wang tile in position $(i,j)$.
See Figure~\ref{fig:coding} for an illustration.

\begin{figure}[htb]
\includegraphics[scale=0.38]{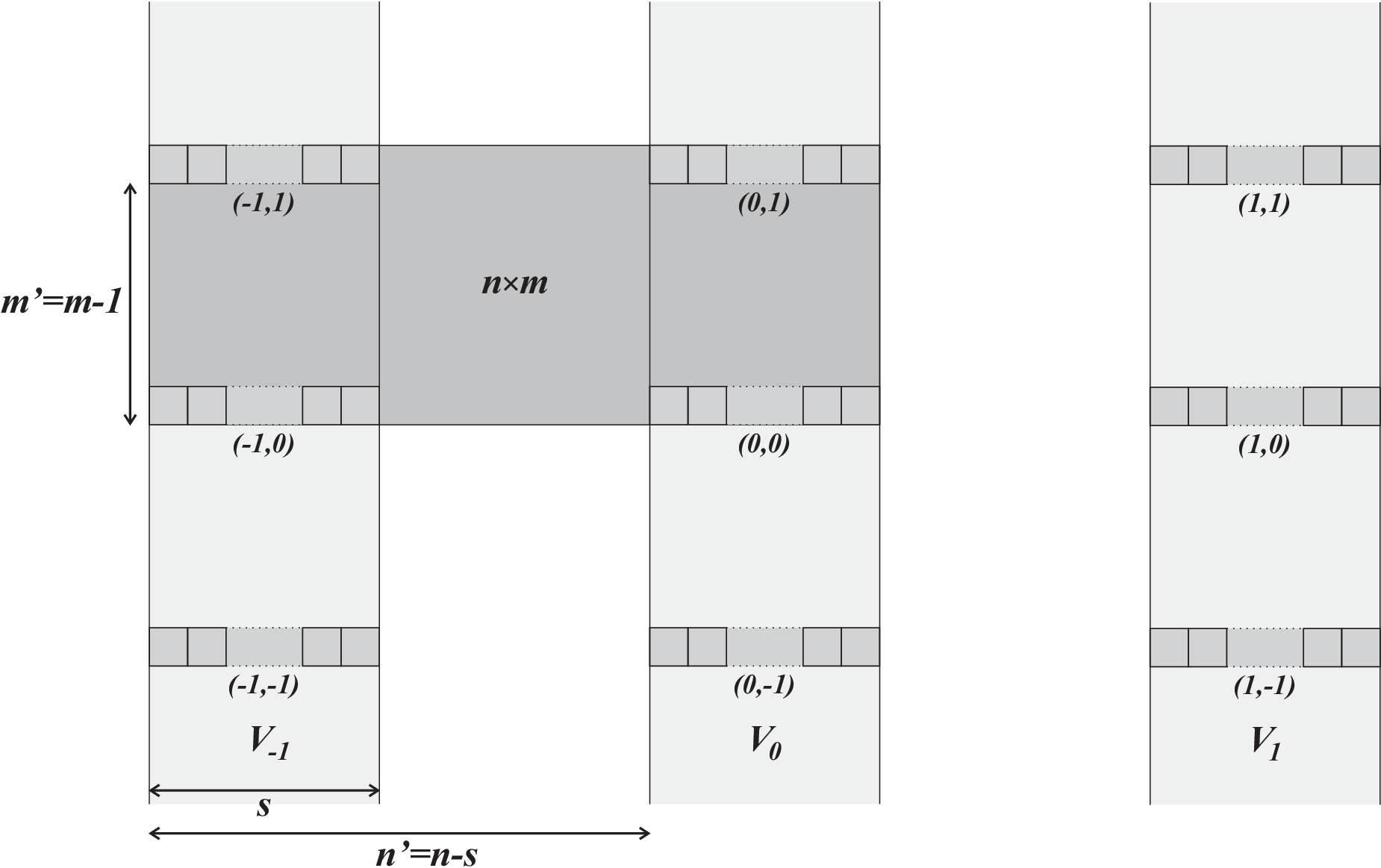}
%
%
\caption{The positioning of the horizontal $s$-bit encodings of Wang tiles in coding $\beta$.
The given coordinates indicate the positions in the Wang tiling that are encoded in the corresponding bit sequences.
A sample rectangle of size $n\times m$ is depicted in dark shading.
Three consecutive vertical $V_i$ strips are highlighted. }
\label{fig:coding}       
\end{figure}

Let us first count all rectangular $n\times m$ patterns that may appear in $\beta(c)$ for some $c\in T^{\Z^2}$, that is, find an upper bound on the cardinality of the set
$$
Q = \bigcup_{c\in T^{\Z^2}} \Patt{\beta(c)}{D}.
$$
As $n=n'+s$ we have that
for all $j\in\Z$ there is $i\in\Z$ such that
$in'+\llbracket s\rrbracket \subseteq j+\llbracket n\rrbracket$, that is,
every $n\times m$ rectangle on the grid fully intercepts one of the vertical strips $V_i$. Analogously,
the rectangle intercepts a horizontal strip $H_j$ and hence some $I_{i,j}$ is fully contained in the rectangle.
This implies that every pattern in $Q$ contains at least one symbol $1$. On the other hand, $n\leq 2n'-s$ so that an $n\times m$ rectangle can not
intersect with more than two strips $V_i$, and analogously it cannot intersect more than two horizontal strips $H_j$.
This means that there are at most four symbols $1$ in each pattern of $Q$.

Let $p\in Q$ and let $c\in T^{\Z^2}$ be such that $p\in \Patt{\beta(c)}{D}$. Let $E=\vec{u}+D$ be a rectangle containing pattern $p$
in $\beta(c)$. We have the following four possibilities.
\begin{itemize}
\item Suppose that $E$ has a non-empty intersection with two consecutive vertical strips $V_i$ and $V_{i+1}$ and with two consecutive
horizontal strips $H_j$ and $H_{j+1}$. Rectangle $E$ can be positioned in at most $2s$ positions relative to these strips, and there at most $t^4$ choices of the Wang tiles encoded in the intersections of the two horizontal and two vertical strips. This means that there are at most $2st^4$ patterns $p$ that can be extracted this way.
\item Suppose that  $E$ has non-empty intersection with two consecutive vertical strips $V_i$ and $V_{i+1}$ and with only one horizontal strip $H_j$. There are at most $2sm$ ways to position the rectangle and at most $t^2$ choices for the two tiles
    encoded within the block. There are hence at most $2smt^2$ patterns $p$ of this type.
\medskip

\item Symmetrically, if $E$ has non-empty intersection with two consecutive horizontal strips $H_j$ and $H_{j+1}$ and with only one vertical strip $V_j$ then the number of extracted patterns is bounded by $nt^2$.
\item Finally, if $E$ only intersects a single vertical and horizontal strip then $E$ contains a single symbol $1$. There are
at most $nm$ positions for this $1$ inside the $n\times m$ rectangle.
\end{itemize}
Adding up the four cases above gives the upper bound
$$
nm+2st^4+2st^2m+t^2n \leq nm+k(n+m)
$$
for the cardinality of $Q$, where we can choose $k=2st^4$, assuming $t\geq 1$.
This choice of $k$ works by a direct calculation due to $t^2\geq 1$, $st^2\geq 1$ and $n\geq 2$: subtracting the left-hand-side from
the right-hand-side yields
$$
2st^4(n+m)-(2st^4+2st^2m+t^2n)=2st^2(t^2-1)m+t^2(2st^2-1)(n-1)-t^2\geq 0.
$$
Note that constant $k=2st^4=|T|^4\cdot 2^{|T|}$ does not depend on $n$ or $m$ but only on the number of tiles in $T$. Note also that the patterns in $Q$ can be effectively constructed. We have established the following result.

\begin{lemma}
\label{lem:count}
The number of different $n\times m$ patterns that appear in $\beta(c)$ over all $c\in T^{\Z^2}$ is at most $nm+k(n+m)$ for
$k=|T|^4\cdot 2^{|T|}$. These patterns can be effectively constructed for a given $T$.

\end{lemma}

\noindent
\emph{Remark}\
A smaller constant $k$ can be obtained by using
a more succinct representation $\alpha$ of tiles $T$ as numbers. One just needs to encode
tiles as natural numbers whose differences $a-b$ identify uniquely $a$ and $b$, so that Lemma~\ref{lem:unique} is satisfied.
Instead of the exponentially growing sequence of representatives $0,1,3,7,15,\dots$  that we use here one can use, for example,
numbers of the Mian-Chowla sequence $1, 2, 4, 8, 13, 21, 31, \dots$ (sequence A005282 in \cite{OEIS}) that only grows polynomially. Then constant $k$
will be bounded by a polynomial of $|T|$.
\medskip

Let us further limit the allowed patterns by removing from $Q$ patterns that contain two $1$'s whose relative positions indicate
neighboring Wang tiles whose colors do not match. More precisely, let $p\in Q$.
\begin{itemize}
\item[(H)] Suppose $p$ contains on some row two symbols $1$, in columns $i$ and $j$, for $i<j$.
In order for $p$ to appear in $\beta(c)$ for some valid tiling $c$ we necessarily must have that
the two symbols $1$ are the contributions of two matching horizontally neighboring tiles in $c$,
so that $i=k+\alpha(a)$ and $j=k+n'+\alpha(b)$ for some integer $k$ and tiles $a,b\in T$
such that the east color of $a$ is the same as the west color of $b$. Hence we remove $p$ from $Q$
if no matching $a,b$ exist such that $j-i=n'+\alpha(b)-\alpha(a)$.
\item[(V)] Suppose $p$ contains symbol $1$ in some column $i$ of the bottom row and some column $j$
of the top row where $i-s<j<i+s$. Now $p$ can appear in $\beta(c)$
only if the two symbols $1$ are the contributions of two vertically neighboring tiles in $c$,
so that $i=k+\alpha(a)$ and $j=k+\alpha(b)$ for some integer $k$ and tiles $a,b\in T$
such that the north color of $a$ is the same as the south color of $b$. We remove $p$ from $Q$
if no matching $a,b$ exist such that $j-i=\alpha(b)-\alpha(a)$.
\end{itemize}
Let $P$ be the set of patterns in $Q$ that are not removed by the conditions (H) and (V) above.
Set $P$ can be effectively constructed and, since $P\subseteq Q$, the upper bound $|P| \leq nm+k(n+m)$ from Lemma~\ref{lem:count}
holds.

Let us next prove that allowing the patterns in $P$ admits precisely the configurations $\beta(c)$ and all their
translates, for all $c\in T^{\Z^2}$ that are valid Wang tilings.

\begin{lemma}
\label{lem:correctness}
With the notations above,
$$\SFT{P} = \{\tau^{\vec{t}}(\beta(c))\ |\ \vec{t}\in\Z^2 \mbox{ and } c\in \Valid{T}\}.$$
\end{lemma}
\begin{proof}
By the definition of $P$ it is clear that for every valid tiling $c\in \Valid{T}$ the encoded configuration $\beta(c)$
only contains allowed patterns in $P$. Hence the inclusion ``$\supseteq$'' holds.

To prove the converse inclusion, consider an arbitrary configuration $e\in\SFT{P}$, that is, $e\in \{0,1\}^{\Z^2}$ such that
$\Patt{e}{D}\subseteq P$. Every pattern in $P$ contains a symbol $1$ so configuration $e$ must contain a symbol $1$ in every $n\times m$ block.

Let us denote, for any $x,y,z,w\in T$, by $\nelio{x}{y}{z}{w}$ the $n\times m$ binary pattern with exactly four $1$'s, two of which are on the bottom row in
columns $\alpha(x)$ and $n'+\alpha(y)$, and two are on the topmost row in columns $\alpha(z)$ and $n'+\alpha(w)$. In other words, the bit
sequences that encode tiles $x,y,z$ and $w$ are in the four corners of the pattern, as in the dark grey block in
Figure~\ref{fig:coding}. Let us call $\nelio{x}{y}{z}{w}$ a \emph{standard block} if the Wang tiles $x,y,z,w$ match each other in
colors as a $2\times 2$ pattern with
$x,y,z$ and $w$ at the lower left, lower right, upper left and upper right position of the $2\times 2$ pattern, respectively.

Consider now any occurrence of a symbol $1$ in $e$, that is, $\vec{u}\in\Z^2$ such that $e(\vec{u})=1$.
Let us prove that there is a standard block  $\nelio{x}{y}{z}{w}$  in $e$ with this occurrence of $1$ representing Wang tile $x$.
Let $p=\tau^{\vec{u}}(e)_{|D}$
be the $n\times m$ pattern with lower left corner at cell $\vec{u}$, so there is a symbol $1$ at the lower left corner of $p$.
By the definition of $P$, pattern $p$ appears in
$\beta(f)$ for some $f\in T^{\Z^2}$. The structure of $\beta(f)$ implies that there is another symbol $1$ in
pattern $p$ on the same horizontal row, say $i$ position to the right of the lower left corner.
By condition (H) above, $i=n'+\alpha(y)-\alpha(x)$ for some tiles $x,y\in T$
such that the east color of $x$ is the same as the west color of $y$. Because no tile in $T$ matches with itself in color,
we have $x\neq y$ and hence $x$ and $y$ are unique by Lemma~\ref{lem:unique}.

Let $\vec{v}=\vec{u}-(\alpha(x),0)$, and extract the $n\times m$
pattern $q=\tau^{\vec{v}}(e)_{|D}$ located $\alpha(x)$ positions to the left of $p$ in $e$. Pattern $q$
contains symbol  $1$ on the bottom row at columns $\alpha(x)$ and $n'+\alpha(y)$. Pattern $q$ appears in
$\beta(f')$ for some $f'\in T^{\Z^2}$ and therefore, due to the structure of
encoded configurations, $q$ must be $\nelio{x}{y}{z}{w}$ for some $z,w\in T$. Conditions (H) and (V) then ensure that
$x, y, z$ and $w$ match in color with each other to form a valid $2\times 2$ pattern of Wang tiles, so
$q=\nelio{x}{y}{z}{w}$ is a standard block.

We have seen that any occurrence of $1$ in $e$ represents a Wang tile $x$ in the lower left corner of some
standard block $q=\nelio{x}{y}{z}{w}$ in $e$.
With an analogous reasoning we see that the same occurrence of bit $1$
is also encoding  the tile $z'$ at the upper left corner of a standard  block $q'=\nelio{x'}{y'}{z'}{w'}$ in $e$.
In $q'$ the symbol $1$ that represents the Wang tile $w'$ at the upper right corner is the same as
the one that represents tile $y$ at the lower right corner in $q$, so that
$\alpha(x)-\alpha(y)=\alpha(z')-\alpha(w')$. By Lemma~\ref{lem:unique} the tiles are unique so that
$x=z'$ and $y=w'$. (See Figure~\ref{fig:proofpic} for an illustration.)
We have that $q'=\tau^{\vec{v'}}(e)_{|D}$ for $\vec{v'}=\vec{v}-(0,m')$.

Analogously, symbol $1$ in position $\vec{u}$ is also in the lower right and upper right corners of standard blocks
in $e$ that overlap with $q$ and $q'$ in two encoded Wang tiles that by Lemma~\ref{lem:unique} are uniquely identified
as $x$ and $z$ and as $x'$ and $z'=x$, respectively. So also the $n\times m$ blocks in $e$ with lower left corners at cells
$\vec{v}-(n',0)$ and $\vec{v}-(n',m')$ are standard blocks.

As cell $\vec{u}$ is any position containing bit $1$ in configuration $e$, we can repeat the reasoning
on the other corners of the standard blocks. By easy induction we see that
$\tau^{\vec{v}_{i,j}}(e)_{|D}$ is a standard block for all $i,j\in\Z$ where $\vec{v}_{i,j}=\vec{v}+(in',jm')$.
We now take $c\in T^{\Z^2}$ such that $c(i,j)$ is the Wang tile encoded in $e$ position $\vec{v}_{i,j}$, for all $i,j\in\Z$,
that is, the unique $t\in T$ such that $\tau^{\vec{v}_{i,j}}(e)(\alpha(t))=1$.
Clearly $\tau^{\vec{v}}(e)=\beta(c)$. Because standard blocks
correspond to correctly tiled $2\times 2$ blocks of Wang tiles we have that
$c\in \Valid{T}$.
\end{proof}

\begin{figure}[htb]
\begin{center}
\includegraphics[scale=0.38]{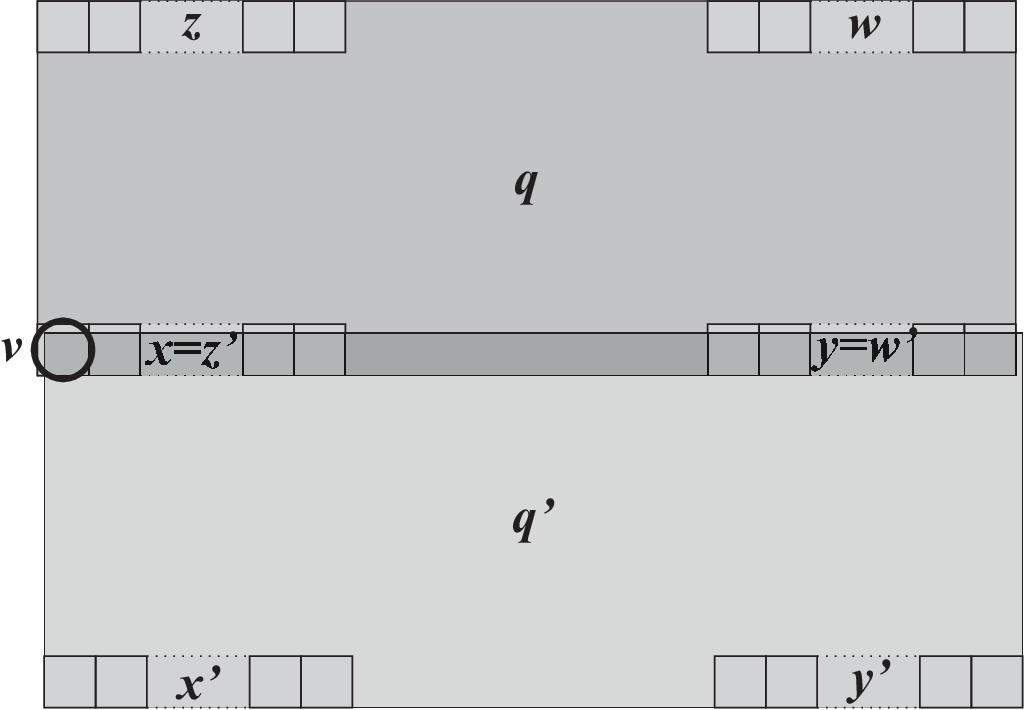}
\end{center}
\caption{Two standard blocks sharing an encoded tile at their lower left and upper left corners, respectively.
The positions of the blocks are uniquely identified by their common row,
as discussed in the proof of Lemma~\ref{lem:correctness}. The circled cell is the position $\vec{v}$ in configuration $e$ in that proof.}
\label{fig:proofpic}       
\end{figure}

We are now ready to prove \cref{thm:main_recoding}.

\begin{proof}[Proof of \cref{thm:main_recoding}]
We first construct an
equivalent tile set $T'$ where no tile matches in color with itself, as shown in the beginning of the section.
We then set $t=|T'|$, $s=2^{t-1}$, $N=3s$ and $k=2st^4$. Let $n\geq N$ and $m\geq 2$ be arbitrary, and let us construct $P$
as above. By Lemma~\ref{lem:count} set $P$ contains at most $nm+k(n+m)$ patterns. By Lemma~\ref{lem:correctness} we have that
$\SFT{P}=\emptyset$ if and only if $\Valid{T}=\emptyset$. Encoding $\beta$ maps periodic configurations to
periodic configurations so also by Lemma~\ref{lem:correctness} there is a periodic configuration in $\SFT{P}$ if and only if
there is a periodic configuration in $\Valid{T}$.
\end{proof}


\section{Conclusions}
\label{sec:conclusions}

We have demonstrated how the low local complexity assumption enforces global regularities in the valid configurations, yielding algorithmic decidability
results. The results were proved in full detail for low complexity configurations with respect to an arbitrary rectangle. The reader can easily verify that
the fact that the considered shape is a rectangle is not used in any proofs presented here, and the only quoted result that uses this fact is Lemma~\ref{lem:Michal2}.
A minor modification in the proof of Lemma~\ref{lem:Michal2} presented in~\cite{szabados} yields that the lemma remains true for any
two-dimensional configuration that has low complexity with respect to any convex shape. We conclude that also \cref{thm:periodic}, \cref{cor:corperiodic}, \cref{cor:cordecidable} and \cref{cor:corminimal} remain true if
we use any convex discrete shape in place of a rectangle.

If the considered shape is not convex the situation becomes more difficult.
Theorem~\ref{thm:periodic} is not true for an arbitrary shape in place of the rectangle but
all counter examples we know are based on periodic sublattices~\cite{cassaigne,karimoutot}. For example, even lattice cells may form a configuration that is horizontally but not vertically
periodic while the odd cells may have a vertical but no horizontal period. Such a non-periodic configuration may be uniformly recurrent and have low complexity with respect to
a scattered shape $D$ that only sees cells of equal parity. It remains an interesting direction of future study to determine if a sublattice structure is the only
way to contradict  Theorem~\ref{thm:periodic} for arbitrary shapes. We conjecture that Corollaries~\ref{cor:corperiodic} and \ref{cor:cordecidable}
hold for arbitrary shapes, that is, that there does not exist a two-dimensional low complexity aperiodic SFT. A special case of this is the recently solved
periodic cluster tiling problem~\cite{bhattacharya,szegedy}.

\cref{cor:goles_1} naturally raises the question whether the additive term $f(n)n$ can be replaced by some constant,
or can at least $f(n)$ in it be replaced by a constant. By \cref{cor:cordecidable} we know that constant $c=0$ does not work, but some other constant might work.

Naturally, there exists $m,n,c$ and a set $P$ of $nm+c$ allowed patterns such that $\SFT{P}$ is aperiodic (take for example any aperiodic SFT and choose $c$ accordingly).
It is not known what might be the smallest such $c$, and this constitutes an interesting open problem. We just know that $c>0$ by \cref{cor:corperiodic}.



\bibliography{biblio}

\end{document}